\documentclass[11pt]{article}
\usepackage[backend=biber, maxbibnames=99]{biblatex}
\usepackage{amsmath,amsfonts,amssymb,amsthm,  mathrsfs,mathtools}
\usepackage[margin=1in]{geometry}
\usepackage{hyperref}
\hypersetup{colorlinks}
\usepackage[ruled]{algorithm2e}
\usepackage{subcaption}
\usepackage{xfrac}
\usepackage{enumitem}
\usepackage[T1]{fontenc}

\usepackage{mleftright}

\usepackage{booktabs}
\usepackage{multirow}

\theoremstyle{plain}
\newtheorem{theorem}{Theorem}[section]
\newtheorem{lemma}[theorem]{Lemma}
\newtheorem{corollary}[theorem]{Corollary}

\theoremstyle{definition}

\newtheorem{observation}[theorem]{Observation}

\DeclareMathOperator*{\argmin}{arg\,min}
\DeclareMathOperator*{\poly}{poly}

\newcommand{\mininfprob}{\textsc{MinInfEdge}}
\newcommand{\mininfprobnode}{\textsc{MinInfNode}}
\newcommand{\mininfprobCL}{\textsc{MinInf-CL}}
\newcommand{\mininfprobnodeCL}{\textsc{MinInfNode-CL}}

\begin{document}

\title{Controlling Epidemic Spread using Probabilistic Diffusion Models on Networks}

%


\author{Amy Babay\thanks{University of Pittsburgh.
Email: \href{mailto:babay@pitt.edu}{babay@pitt.edu}} \and Michael Dinitz\thanks{Johns Hopkins University. 
Email: \href{mailto:mdinitz@cs.jhu.edu}{mdinitz@cs.jhu.edu}} \and Aravind Srinivasan\thanks{University of Maryland, College Park. 
Email: \href{mailto:srin@cs.umd.edu}{srin@cs.umd.edu}} \and  Leonidas Tsepenekas \thanks{University of Maryland, College Park.  
Email: \href{mailto:ltsepene@umd.edu}{ltsepene@umd.edu}}
\and Anil Vullikanti \thanks{University of Virginia.  
Email: \href{mailto:vsakumar@virginia.edu}{vsakumar@virginia.edu}}
}

\date{}

\maketitle

\begin{abstract}
The spread of an epidemic is often modeled by an SIR random process on a social network graph. The \mininfprob{} problem for optimal social distancing involves minimizing the expected number of infections, when we are allowed to break at most $B$ edges; similarly the \mininfprobnode{} problem involves removing at most $B$ vertices. These are fundamental problems in epidemiology and network science. While a number of heuristics have been considered,  the complexity of these problems remains generally open. In this paper, we present two bicriteria approximation algorithms for \mininfprob{}, which give the first non-trivial approximations for this problem. The first is based on the cut sparsification result of Karger \cite{karger:mathor99}, and works when the transmission probabilities are not too small. The second is a \emph{Sample Average Approximation (SAA)} based algorithm, which we analyze for the Chung-Lu random graph model. We also extend some of our results to tackle the \mininfprobnode{} problem.
\end{abstract}

\section{Introduction}\label{sec:intro}
With the COVID-19 pandemic and future such pandemics in mind, \emph{computational epidemiology}, powered by AI and efficient algorithms, has emerged as a vital discipline. There are two major sources of uncertainty in typical applications of computational epidemiology: how the disease will unfold probabilistically (we may have a good model for this, but have limited control over such stochasticity), and models for contact between members of a population (social-contact networks). In this work, we take a rigorous stochastic-optimization approach to develop provably-good approximation algorithms for budgeted epidemic control under such sources of uncertainty.

The most widespread tool for modeling the spread of an epidemic in a social-contact network $G=(V,E)$ are SIR random processes \cite{pastor2015epidemic,marathe:cacm13}. According to those, the infection starts at a given set $S$ of vertices, where without loss of generality we can assume that these nodes are merged into a single infectious vertex $s$ (this is formally explained in Section \ref{sec:def}). Afterwards, if any node $u \in V$ gets infected, it spreads the disease independently to each ``healthy'' neighbor $v \in V$ of it with probability $p_{u,v}$---also denoted $p_e$---where $e = (u,v) \in E$ is an edge of the given contact network.

In order to control and mitigate the spread of the disease, there are two primary interventions studied in the literature. The first involves \emph{social distancing}, and is modeled as removing a subset $F \subseteq E$ of edges from the graph. The second corresponds to \emph{vaccination}, and is modeled by removing a set $V'\subseteq V$ of nodes. 
There is a significant cost associated with removing nodes and edges, and this motivates the problems we study in this paper. In the \mininfprob{} problem, the goal is to choose a set $F$ of edges for social distancing, so that the cost of $F$ is at most some budget $B$, and the expected number of infections is the minimum possible; similarly, \mininfprobnode{} involves removing a subset of at most $B$ nodes to minimize the expected number of infections.

Despite the significant importance of the \mininfprob{} and \mininfprobnode{} problems, they both remain quite open. A number of heuristics have been proposed, which choose edges or nodes based on local structural properties, e.g., degree, centrality and eigenvector components. However, these do not give any guarantees in general, except in very special random graph models, e.g.,~\cite{bollobas+errorattack04}. The only prior work on \mininfprob{} with rigorous guarantees is for the case of deterministic graphs, where we assume $p_{u,v}=1$ for all $(u,v) \in E$~\cite{hayrapetyan2005,DBLP:conf/soda/EubankKMSW04,ekmsw-2006,svitkina2004}. This scenario actually models a highly-contagious disease, and can be viewed as the SI model~\cite{marathe:cacm13}. In this paper, we obtain the first rigorous results for both \mininfprob{} and \mininfprobnode{}.


\subsection{Formal Problem Definition}\label{sec:def}
Suppose we have an undirected graph $G=(V,E)$ with edge weights $c_e \geq 0$ for every edge $e \in E$ (representing the cost of removing the social connection $e$). Let $\Delta$ denote the maximum degree of any node in $G$, and $n= |V|$, $m=|E|$.
  
We assume an SIR model of disease spread, in which each node is in one of the states S (susceptible), I (infectious) or R (recovered). We also assume the infection starts at
a subset $I_0 \subseteq V$.
An infectious vertex $v$ infects each susceptible neighbor $u\in N(v)$ once, independently with probability $p_{u,v}\in[0, 1]$, where $N(v) = \{u \in V: (u,v)\in E\}$. This is equivalent to a percolation process~\cite{pastor2015epidemic, marathe:cacm13}: consider a random subgraph $G(\vec p) = (V,E(\vec p))$ obtained by retaining each edge $e \in E$ independently with probability $p_e$ (and thus removing each edge with probability $1-p_e$).
In particular, the probability that a set $V_{inf}$ of vertices is reachable from $I_0$ in $G(\vec p)$ is precisely equal to the probability that the set $V_{inf}$ becomes infected during the SIR process.
We will sometimes abuse notation and let $G(\vec p)$ also represent the distribution over subgraphs thus obtained. Without loss of generality, we can assume that $I_0$ consists of a single vertex $s$, since we can add a meta-vertex $s$ with edges to all vertices in $I_0$ with probability 1. Finally, some of our results assume a uniform probability setting in which $p_e=p$ for all $e\in E$; in this case we denote the random graph $G(\vec{p})$ by $G(p)$.

A social distancing strategy corresponds to the removal of a subset $F\subseteq E$ of edges; for such an $F$, we denote by $inf(V,E \setminus F,s)$ the number of vertices that are in the same connected component as $s$ in the \emph{residual} graph $G_F = (V,E \setminus F)$. For simplicity, we refer to $F$ as a cut or a cut-set, though it need not always induce a cut in the graph for our problems of interest. The expected number of infected vertices in the percolation process is then $\mathbb{E}_{G(\vec p)}[inf(V,E(\vec p) \setminus F,s)]$. 

For the vaccination intervention, $F\subseteq V$ is a set of nodes to be removed, and in this case, $G_F = (V, E-\{(u, v)\in E: u\in F\mbox{ or }v\in F\})$ is the subgraph obtained by removing edges incident to nodes in $F$; here $inf(V,E \setminus \{(u, v)\in E: u\in F\mbox{ or }v\in F\},s)$ will denote the number of infected vertices when the edges incident to vertices in $F$ are removed. For the expected number of infections in the percolation process we use $\mathbb{E}_{G(\vec p)}[inf(V,E(\vec p) \setminus \{(u, v)\in E(\vec p): u\in F\mbox{ or }v\in F\},s)]$.

\paragraph{The \mininfprob{} Problem:}
Besides the already-described input, we are given a budget $B$, and the goal is to choose a set $F \subseteq E$ of edges such that:

\smallskip \noindent 
(a) $c(F) = \sum_{e \in F}c_e \leq B$, i.e., the total cost of the set $F$ of edges to be removed, is at most $B$.

\smallskip \noindent 
(b) $\mathbb{E}_{G(\vec p)}[inf(V,E(\vec p) \setminus F,s)]$, i.e., the expected number of nodes reachable from $s$ when we remove the edges in $F$ and conduct the disease percolation on the remaining graph, is minimized. 

\paragraph{The \mininfprobnode{} Problem:}
Besides the already-described input, we are given a budget $B$, and the goal is to choose a set $F \subseteq V$ of vertices such that:

\smallskip \noindent 
(a) $|F| \leq B$, i.e., the total number of removed vertices is at most $B$.

\smallskip \noindent 
(b) $\mathbb{E}_{G(\vec p)}[inf(V,E(\vec p) \setminus \{(u, v)\in E(\vec p): u\in F\mbox{ or }v\in F\},s)]$ is minimized. \\

\noindent
\textbf{$(\alpha, \beta)$-approximation.} As in~\cite{hayrapetyan2005,DBLP:conf/soda/EubankKMSW04}, we will focus on bicriteria approximation algorithms. Here we define this notion only for \mininfprob{}, since the case of \mininfprobnode{} is almost identical. We say that a solution $F \subseteq E$ is an $(\alpha,\beta)$-approximation if $c(F) \leq\alpha B$, and $\mathbb{E}_{G(\vec p)}[inf(V,E(\vec p) \setminus F,s)]\leq \beta \cdot \mathbb{E}_{G(\vec p)}[inf(V,E(\vec p) \setminus F^*,s)]$, where $F^*$ is an optimal solution for the given instance.

\paragraph{Random Models for Networks:} It is well-recognized that with the ever-growing importance of networks and network science, we need good random-graph models for predictive applications, simulations, testing of new algorithms etc.: see, e.g., \cite{Barabasi509,bollobas+errorattack04}. 

In our context of social-contact networks, the random-graph model of Chung and Lu~\cite{DBLP:journals/siamdm/ChungL06} is particularly useful. In this model, we have a set of vertices $V$, and a weight $w_v$ for every node $v \in V$ that denotes its expected degree in the graph; let $w_{min}= \min_v w_v$ and $w_{max} = \max_v w_v$. The edges $E$ of the graph are determined via the following random process. For every $u,v \in V$, the probability of having the $(u, v)$ edge in $E$ is
\begin{align}
    q_{u,v} = \frac{w_u w_v}{\sum_{r\in V} w_r},   \notag
\end{align} 
where these edges are present independently and self-loops are allowed. A natural assumption here is $w_{min} = O(1)$. A common instantiation of this model is with a  power law, in which $n_i$, the number of nodes of weight $i$, satisfies $n_i = \Theta(n/i^{\beta})$, with $\beta > 2$ being a model parameter. In our paper, we are using the power law instantiation every time we consider this model.

The random graphs captured by the Chung-Lu model are more realistic than those of the simple Erd\H{o}s-Renyi model~\cite{DBLP:conf/soda/EubankKMSW04}. The reason for this is imposing a specified degree sequence that models the heavy-tailed nature of real-world degree distributions. 

We refer to \mininfprob{} and \mininfprobnode{} when the graph $G = (V,E)$ is from the Chung-Lu model as \mininfprobCL{} and \mininfprobnodeCL{}, respectively.
The random process for constructing the graph $G=(V,E)$ in the Chung-Lu model should not be confused with the percolation process occurring on $G$ during the spread of the disease.
In the case of \mininfprobCL{} and \mininfprobnodeCL{}, the reader can view the whole process as happening in two steps. At first, $G=(V,E)$ is chosen randomly according to the Chung-Lu model. Afterwards, the disease starts its diffusion in the chosen network according to the probability vector $\vec p$.

\subsection{Contributions and Outline}

We mostly focus on \mininfprob{} for the rest of the paper. In Appendix~\ref{sec:vacc} we discuss which of our results extend to the case of \mininfprobnode{}.

In Section~\ref{sec:karger}, we present an $(O(1), O(1))$-approximation for unit edge-cost \mininfprob{} (all edges of $G$ have cost $1$). This result is for the uniform $p$ probability setting, in the regime where Karger's cut sparsification result holds.
Let $\hat{G}$ be the weighted graph obtained by setting the weight $w_e$ of each edge $e$ equal to $p$, and let $\hat{c}$ denote the weight of the minimum cut in $\hat{G}$.
Karger's result (Theorem~\ref{thm:karger}) states that if
$\hat{c}\geq 9\ln{n}$, then the size of every cut in $G(p)$ is close to the corresponding cut in $\hat{G}$. In this case, we are able to reduce \mininfprob{} to a problem from~\cite{hayrapetyan2005}, using just one random sample $G(p)$.
However, even this setting is not trivial, and to the best of our knowledge, this is the first result when the transmission probability is not 1. 

In Section \ref{sec:SAA} we present a sampling framework for \mininfprob{} that utilizes the powerful sample-average-approximation (SAA) approach \cite{KleywegtSH01,RuszczynskiS03,Shapiro03,swamy2012sampling}. Specifically, we sample a polynomial number of graphs from $G(\vec p)$ and then formulate a linear program (LP) that describes the empirical estimate of the optimal solution of those samples. Afterwards, we solve this LP and provide a randomized-rounding procedure that transforms its fractional solution into an integral one. Let $F_0$ be the solution (set of edges to remove) that we compute, $OPT$ the value of the optimal solution, and $\Gamma$ the expected number of simple paths\footnote{``Paths" will refer throughout to \emph{simple} paths: ones in which no nodes or edges are repeated.} in a randomly drawn graph from $G(\vec p)$, where the randomness also includes the random choice of $G$, in case $G$ is drawn from a random-graph model.

\paragraph{Three different sources of randomness/uncertainty:} Our statements will refer to (combinations of) three distinct sources of randomness:
\begin{itemize}
    \item \textbf{Type 1:} This randomness is over the random choice, if any, of our network $G = (V,E)$ (such as randomness resulting from choosing $G$ according to the Chung-Lu model). If the network $G$ is deterministic, Type 1 is vacuous: there is no randomness. 
    \item \textbf{Type 2:} This randomness arises from the choices of our randomized rounding algorithm.
    \item \textbf{Type 3:} This type of randomness refers to the random percolation/diffusion of the disease, governed by $\vec{p}$.
\end{itemize}

Our main theorem for the SAA approach of Section \ref{sec:SAA} is summarized in the following, where ``$\log$'' denotes the natural logarithm throughout. 

\begin{theorem}
\label{thm:main}
For any chosen constants $\epsilon > 0$ and $\gamma > 1$, the following hold:
\begin{itemize}
\item with probability at least $1 - O(n^{-\gamma})$, where the randomness is solely of Type 2, we have $c(F_0) \leq O(\frac{\gamma}{\epsilon}) \log{n} \cdot B$; 
\item there exists an event $\mathcal{A}$ with $\Pr[\mathcal{A}] \geq 1-O(\frac{1}{n^2})-O\big(\frac{\Gamma \log n}{ \epsilon^2 n^{\gamma}}\big)$ and $\mathbb{E}[inf(V,E(\vec p) \setminus F_0,s) \bigm| \mathcal{A} ]\leq (1+\epsilon) \cdot OPT$. Here, randomness is with respect to  Type 1 (if applicable), Type 2, and Type 3.
\end{itemize}
\end{theorem}

Observe now that if $\Gamma \leq  \poly(n)$\footnote{Throughout, ``poly" will denote an arbitrary univariate or bivariate polynomial.}, we can choose $\gamma$ to be large enough, such that $\Pr[\mathcal{A}] \geq 1-O(1/n^2)$. As we show in Section \ref{sec:SAA} this immediately implies the following corollary.

\begin{corollary}
When $\Gamma \leq  \poly(n)$, we have $\mathbb{E}[inf(V,E(\vec p) \setminus F_0,s)]\leq (1+O(\epsilon) + O(1/n)) OPT$, where the randomness is with respect to  Type 1 (if applicable), Type 2, and Type 3.
\end{corollary}

Hence, in Section~\ref{sec:path-count} we prove that a family of Chung-Lu random-graphs satisfies the $\Gamma \leq  \poly(n)$ property (recall this model captures realistic social-contact networks well~\cite{DBLP:conf/soda/EubankKMSW04,ekmsw-2006}). Under this property, are main result informally says that \emph{we can approximate the budget to within a factor $O(\log n)$ with high probability, and the expected number of infected people to within a constant factor}. 

A remark regarding Section \ref{sec:SAA} is that our goal is to present a "proof of concept", so we do not optimize the constants in our algorithms, and we are content with polynomial running times. In particular, we do not spell out the actual running times of our algorithms: these will easily be seen to be bounded by polynomials of $n$ and $m$. We remark that most of the prior work on this problem has been experimental, and that our paper is the first to give rigorously-proven results.

Section~\ref{sec:path-count} develops the above-mentioned $\poly(n)$ bound on $\Gamma$ for a realistic Chung-Lu family of graphs. More generally, it shows a \emph{phase-transition} phenomenon for the expected number of paths of any length $k$, as a function of the model parameter $\beta$: this is proved to be at most $\poly(n,2^k)$ for $\beta > 3$, and to be at least  $(k!)^{\Omega(1)}$ for $\beta < 3$. This leads to our provably-good approximation algorithms for the Chung-Lu family of graph models when $\beta > 3$.
 
In Section~\ref{sec:saa-general}, we show a slightly different SAA approach combined with a deterministic rounding which gives an $(O(n^{2/3}), O(n^{2/3}))$-approximation for general graphs.
 


\subsection{Further Related Work}
There has been much work on heuristics for interventions for the SIR model~\cite{YANG2019115, EAMES200970, PhysRevLett.91.247901,Miller2007EffectiveVS,Barabasi509, sambaturu2020designing}. In particular, heuristics based on degree or centrality, e.g., \cite{PhysRevLett.91.247901,Miller2007EffectiveVS}, have been shown to be quite effective in many classes of networks (including random graphs), but these do not provide any guarantees. The work of \cite{sambaturu2020designing} explores the use of the sample average approximation method, but has worst-case approximation bounds as large as $O(n)$.

However, as mentioned earlier, rigorous results are only known for the setting where $p_{u,v}=1$ for all $(u,v) \in E $~\cite{bollobas+errorattack04,hayrapetyan2005,ekmsw-2006,svitkina2004}. The \mininfprob{} problem is known to be NP-hard even in this setting \cite{hayrapetyan2005, svitkina2004}, and constant factor bicriteria approximation algorithms are known.

We note that another related direction of work has been on reducing the first eigenvalue, referred to as the spectral radius, based on a characterization of the time to die out in SIS models (in which, unlike the SIR model, an infected node switches back to state S)  \cite{ganesh+topology05}. There has been much work on reducing the spectral radius, e.g.,
\cite{PreciadoVM13_2,PreciadoVM13,PreciadoVM14,SahaSDM15,Ogura2017}. However, these results do not imply any guaranteed bounds for \mininfprob{} or \mininfprobnode{}.

\section{\mininfprob{} with Unit Edge-Costs and Uniform Probabilities}
\label{sec:karger}

In this section we are going to consider a special case of \mininfprob{}. Specifically, we assume that the edge costs of the network $G=(V,E)$ are all $1$, i.e., $c_e = 1$ for all $e \in E$. Moreover, we will work under he uniform transimition probability setting.

For a random graph $G(p) = (V,E(p))$ and any $F \subseteq E$, let $F(p) = F \cap E(p)$ be the random cut corresponding to $F$ in $G(p)$. Let also $c_{min}$ be the size of the smalletst cut in $G$. We are going to use a cut sparisification result of Karger \cite{karger:mathor99}.

\begin{theorem}\label{thm:karger}
(\cite{karger:mathor99}) Let $\epsilon = \sqrt{ \frac{3(d+2)(\ln n)}{c_{min} \cdot p}}$ for some $d > 0$.  If $\epsilon \leq 1$ then, with probability at least $1 - O(1/n^d)$, we have $\big{|}|F(p)| - \mathbb{E}_{G(p)}[|F(p)|]\big{|} \leq \epsilon \mathbb{E}_{G(p)}[|F(p)|]$ for every $F \subseteq E$.
\end{theorem}

\begin{observation}\label{thm:reg-karg}
When $c_{min} \cdot p \geq 9\ln n$, the statement of Theorem \ref{thm:karger} holds with high probability, i.e., with probability at least $1-O(1/n)$.
\end{observation}

Observation \ref{thm:reg-karg} basically determines the regime where the results of this section hold. However, notice that $c_{min} \cdot p \geq 9\ln n$ is a realistic assumption, since for most real-life scenarios the transmission probability will be some constant, and the size of the minimum cut in $G$ can very well be $\Omega(\ln n)$.


To tackle \mininfprob{} in the current setting, we are going to reduce it to a problem from \cite{kempe:esa05}, namely the \textsc{Minimum-Size Bounded-Capacity Cut} problem (MinSBCC).  In this problem, we are given a graph $G = (V,E)$, a source vertex $s \in V$, and a budget $B$. We are then asked to find a set $F \subseteq E$ of at most at most $B$ edges, which minimizes the number of nodes in the same component as $s$ in $G_F = (V, E \setminus F)$, i.e., $inf(V, E \setminus F, s)$. The main result of \cite{kempe:esa05} follows.

\begin{theorem} \label{thm:minsbcc}
For any $\lambda \in (0,1)$, there exists a polytime $(\frac{1}{\lambda}, \frac{1}{1-\lambda})$-approximation algorithm for MinSBCC: it finds a cut of size at most $\frac{1}{\lambda} B$, in which the number of nodes in the same component as $s$ in the resulting subgraph is at most $\frac{1}{1-\lambda}$ times the value of the optimal solution with size $B$.
\end{theorem}





Our approach for solving \mininfprob{} goes as follows. At first, we sample a graph $H = (V, E')$ from $G(p)$. Then, we create an instance of MinSBCC, where the graph under consideration is $H$, the source vertex is $s$, and the budget is $\gamma B p$ for a small constant $\gamma$ which we set later. Finally, we run the $(\frac{1}{\lambda}, \frac{1}{1-\lambda})$-approximation of \cite{kempe:esa05} on the created instance of MinSBCC, and get a solution $F' \subseteq E'$. Let now $S$ be all the vertices that are in the same connected component as $s$ in $H_{F'} = (V, E' \setminus F')$. Our returned solution for the original instance of \mininfprob{} is $\bar F = \{\{u,v\} \in E : u \in S, v \not\in S\}$.  


\begin{lemma} \label{lem:min-budget}
When the assumption of Observation \ref{thm:reg-karg} holds, $|\bar F| \leq \frac{\gamma}{(1-\epsilon)\lambda} B$ with probability at least $1-O(1/n)$, where $\epsilon$ is as in Theorem \ref{thm:karger} with $\epsilon \in (0,1)$.
\end{lemma}
\begin{proof}

Notice that the vertices that are in the same connected as $s$ in $(V,E \setminus \bar F)$, are exactly those that are connected to $s$ in $(V, E' \setminus F')$. Therefore, the random cut corresponding to $\bar F$ in $G(p)$ is $F'$, i.e., $F' = \bar F(p)$. Hence, $\mathbb{E}_{G(p)}[F'] = \mathbb{E}_{G(p)}[\bar F(p)] = p |\bar F|$. Therefore, using Theorem \ref{thm:karger}, we have that with probability at least $1-O(1/n)$: $\big{|}|F'| - p |\bar F|\big{|} \leq \epsilon p |\bar F| \implies |F'| \geq (1-\epsilon)p |\bar F|$.

Since $|F'| \leq \frac{\gamma Bp}{\lambda}$ (we ran the algorithm of Theorem \ref{thm:minsbcc} with budget $\gamma Bp$), we get $ \frac{\gamma Bp}{\lambda} \geq |F'| \geq (1-\epsilon)|\bar F|p$ with probability at least $1-O(\frac{1}{n})$.  Rearranging terms implies that $|\bar F| \leq \frac{\gamma}{(1-\epsilon)\lambda} B$.
\end{proof}

\begin{lemma} \label{lem:min-approx}
$|S| \leq \frac{\gamma}{1-\lambda} OPT$ with probability at least $1 - \frac{2}{\gamma}$, where $OPT$ is the value of the optimal solution (the expected number of nodes infected).
\end{lemma}
\begin{proof}
Let $F^*$ denote the optimal solution (so $|F^*| \leq B$), and let $\hat F = F^* \cap E'$ be a random variable denoting the edges of $F^*$ that are present in $E'$. Let also $S_{\hat F}$ be the random variable denoting the nodes that are in the same connected component as $s$ in $(V, E' \setminus \hat F)$.  We say that there was a ``success'' in the process of sampling $H$ if the following two conditions are satisfied: \textbf{1)} $|\hat F| \leq \gamma Bp$ and \textbf{2)} $|S_{\hat F}| \leq \gamma \cdot OPT$. If either condition is not true we say that there was a ``failure''.  

Suppose that there was a success. Then the first condition implies that $\hat F$ was a feasible solution for the MinSBCC instance (since its size was within the given budget), and hence $|S| \leq \frac{1}{1-\lambda} |S_{\hat F}|$.  Then the second condition implies $|S| \leq \frac{\gamma}{1-\lambda} OPT$ as desired.

Finally, we need to show that the probability of success is at least $1 - \frac{2}{\gamma}$, or equivalently that the probability of failure is at most $\frac{2}{\gamma}$. Clearly $\mathbb{E}_{G(p)}[|\hat F|] = p|F^*| \leq pB$, so by Markov's inequality $\Pr[|\hat F| > \gamma Bp] \leq \frac{1}{\gamma}$. Similarly, $\mathbb{E}[|S_{\hat F}|] = OPT$ by the definition of $OPT$, and so by Markov $\Pr[|S_{\hat F}| > \gamma \cdot OPT] \leq \frac{1}{\gamma}$. A union bound implies that the probability of failure is at most $2/\gamma$.
\end{proof}


\begin{theorem} \label{thm:min-main}
There exists an $(O(1), O(1))$-approximation for the \mininfprob{} that works with high probability, as long as the assumption of Observation \ref{thm:reg-karg} holds.
\end{theorem}

\begin{proof}
If we set $\gamma$ to be a large enough constant (say, $4$), then with probability at least $1/2 - O(1/n)$ we return a solution $\bar F$ which violates the budget by at most $O(1)$ (Lemma~\ref{lem:min-budget}), and the size of the connected component in $(V,E \setminus \bar F)$ which contains $s$ is at most $O(1) \cdot OPT$ (Lemma~\ref{lem:min-approx}).  Clearly this implies that $\mathbb{E}_{G(p)}[inf(V,E(p) \setminus \bar F, s)]$ is also at most $O(1) \cdot OPT$. Thus, our algorithm gives the bounds in Theorem~\ref{thm:min-main} with constant probability. By repeating this process $O(\log n)$ times and taking the best single solution, this algorithm can be made to work with high probability.
\end{proof}

\section{The SAA Path-Dependent Framework for Arbitrary Networks}\label{sec:SAA}

Consider a general instance of \mininfprob{}. For a suitable number $N \leq \poly(n,m)$ that is going to be set later, we simulate the disease-percolation process on $G$ independently $N$ times. In other words, we independently sample $N$ graphs $G_j=(V, E_j)$, $j=1, 2, \ldots, N$, where $E_j \subseteq E$ is the subset of edges acquired in the $j^{th}$ simulation (or sample), when each edge is retained with probability $p_e$. The heart of our approach is to then show how these ``typical'' samples $G_j$ can guide us towards computing a provably-good solution for our given probabilistic percolation model. 

We start by presenting the linear program LP (\ref{LP-1})-(\ref{LP-4}). This LP models an ``empirical'' solution to the problem, when the diffusion process can only result in the graphs $G_j$, and each of these graphs materializes with probability $1/N$. We use $\mathcal{P}(s, v, G_j)$ to denote the set of paths from $s$ to $v$ in the graph $G_j$, and $[k]$ to denote the set $\{1, 2, \ldots, k\}$ for any positive integer $k$. For the integral version of our LP, $x_e$ is the indicator variable for removing edge $e$, and $y_{vj}$ the indicator for vertex $v$ \emph{not becoming reachable} from $s$ in $G_j$ after our edge-removal. Then, constraint (\ref{LP-2}) makes sure that $v$ is disconnected from $s$ in $G_j$ iff for every path of $\mathcal{P}(s, v, G_j)$ at least one edge of the path has been removed. Constraint (\ref{LP-3}) captures the budget constraint, and the objective function (\ref{LP-1}) measures exactly the expected number of infections, when each $G_j$ appears with probability $1/N$. Finally, in order to be able to efficiently solve the system, the $\{0,1\}$--variables are relaxed to lie in $[0,1]$.
\begin{align}
    \min &\frac{1}{N}\sum_{j \in [N]} \sum_{v \in V} (1-y_{vj}) \text{ such that} \label{LP-1} \\
    &\sum_{e\in P} x_e \geq y_{vj}, ~~\forall j \in [N], ~\forall v \in V, ~\forall P\in \mathcal{P}(s, v, G_j)  \label{LP-2}\\
    &\sum_{e \in E} c_e x_e \leq B \label{LP-3}\\
    &x_e, y_{vj} \in [0,1], ~\text{ for all $j\in [N], v \in V, e \in E$}\label{LP-4}
\end{align}

Our algorithm involves the following steps:
\begin{enumerate}
\item Solve LP (\ref{LP-1})-(\ref{LP-4}), and let $x, y$ be the optimal fractional solution. This solution can be computed in polynomial time via the ellipsoid method, with a \emph{separation oracle} that checks if the shortest-path distance from $s$ to $v$ in $G_j$ (with edge weights $x_e$) is less than $y_{vj}$ \cite{gls}. 
\item For some user-specified constant $\epsilon \in (0,1)$, define the following sets for the sake of analysis:
\[ 
    S(j)=\{v \in V: y_{vj}\geq \epsilon\} \text{ for every } j \in [N], \mbox{ and }
    \mathcal{P}_{hit}=\cup_j \cup_{v\in S(j)} \mathcal{P}(s, v, G_j). 
\] 
\item Let $F_0$ denote the set of edges which will constitute our returned solution. For some constant $\gamma$ that will be defined later, put each edge $e \in E$ independently in $F_0$, with probability 
\begin{align}
    x'_e=\min\Big{\{}\frac{(\gamma+5)x_e\log{n}}{\epsilon}, 1\Big{\}} \notag
\end{align}
\end{enumerate}

For any fixed $F \subseteq E$, we define random variables $h(G_j,F)$ and $h(G,F)$, where the randomness here is over the choice of the $G_j$'s, i.e., the randomness is of Type 3. Let $h(G_j, F) = inf(V,E_j \setminus F,s)$ and $h(G, F) = \frac{1}{N}\sum_{j=1}^N h(G_j, F)$; the former represents the number of infections in the $j$-th sample if $F$ are the edges to be removed, and the latter represents the average number of infections over the $N$ sampled graphs if again $F$ is the set of edges removed. 

For the small user-defined constant $\epsilon > 0$, we now choose $N = \frac{3n}{\epsilon^2} \log \big{(}n^2 \cdot 2^{m+1}\big{)}$ and present a simple concentration result in Lemma~\ref{lemma:conc}; note that for this choice we have $N = \poly(n,m)$ and hence our algorithm runs in polynomial time.

\begin{lemma}\label{lemma:conc}
For the chosen value $N=\frac{3n}{\epsilon^2} \log \big{(}n^2 \cdot 2^{m+1}\big{)}$, with probability at least $1-\frac{1}{n^2}$, we have $h(G, F)\in \big{[}(1-\epsilon)\mathbb{E}[h(G, F)], (1+\epsilon)\mathbb{E}[h(G, F)]\big{]}$ for all sets $F\subseteq E$. The expectation here is over randomness of Type 3, and specifically over the random sampling of the $N$ graphs $G_j$.
\end{lemma}

\begin{proof}
 For a fixed set $F$, the quantities $h(G_j, F)$ are independent. Further, let $X_j = \frac{h(G_j, F)}{n} \in [0, 1]$ and $X = \sum^{N}_{j=1}X_j$. Note that $X = \frac{N}{n}h(G,F)$. Using the Chernoff bound of Lemma \ref{chernoff} yields:
\begin{align}
    \Pr\Big{[}X \notin \big{[}(1-\epsilon)\mathbb{E}[X], (1+\epsilon)\mathbb{E}[X]\big{]} \Big{]} & \leq 2 e^{-\frac{\epsilon^2 \mathbb{E}[X]}{3}} = 2 e^{-\frac{\epsilon^2 N \cdot \mathbb{E}[h(G,F)]}{3n}} \leq \frac{1}{n^2 2^m} \notag
\end{align}
To get the last inequality we use the definition of $N$, and the fact that $\mathbb{E}[h(G,F)] \geq 1$ (since there is always at least one infection, namely the node $s$). Finally, since $X = \frac{N}{n}h(G,F)$, we also have:
\begin{align}
      \Pr\Big{[}h(G,F) \notin \big{[}(1-\epsilon)\mathbb{E}[h(G,F)], (1+\epsilon)\mathbb{E}[h(G,F)]\big{]} \Big{]}  \leq \frac{1}{n^2 2^m} \notag
\end{align}
Because the number of possible subsets $F$ is $2^m$, a union bound over them concludes the proof. 
\end{proof}

Let $F^* = \argmin_F\mathbb{E}[h(G, F)]$, where the expectation is again over the random sampling of the graphs $G_j$ (Type 3 randomness). Since for every $F$ we have $\mathbb{E}[h(G_j,F)] = \mathbb{E}_{G(\vec p)}[inf(V,E(\vec p) \setminus F,s)]$ for all $G_j$, and $\mathbb{E}[h(G,F)] = \frac{1}{N}\sum_{j}\mathbb{E}[h(G_j,F)]$, we see that $F^*$ is actually the optimal edge set for \mininfprob{}. Also, we define the random variable $\hat{F} = \argmin_F h(G, F)$, denoting the optimal integral solution of LP (\ref{LP-1})-(\ref{LP-4}); $\hat F$ is actually the optimal empirical solution for the sampled set of graphs. Recall now that $F_0$ is the subset of edges computed by our LP rounding algorithm, and recall the parameter $\Gamma$ from Section~\ref{sec:intro}, indicating the expected number of paths in a randomly-drawn graph (with randomness being of types 1 and 3). 

\begin{proof} (\textbf{Theorem~\ref{thm:main}.})
Showing the first part of the theorem is easy. Since each edge $e$ is removed (independently) with probability $x'_e$, the expected cost of the removed edges is $$\mathbb{E}[c(F_0)] \leq \sum_e c_e x'_e \leq \frac{(\gamma+5)\log{n}}{\epsilon}\sum_e c_e x_e \leq \frac{((\gamma+5)\log{n})B}{\epsilon}$$ where the last inequality follows from constraint (\ref{LP-3}). Next, we can assume w.l.o.g. that $B = 1$. To do so, we first hard-wire $x_e = 0$ for all edges $e$ with $c_e > B$, thus ignoring these edges in our edge-removal problem. Then, we uniformly scale all remaining $c_e$'s and the budget by a factor of $1/B$. Using the second statement of Lemma \ref{chernoff} with $R = (6(\gamma + 5)\log n) / \epsilon$ gives:
$$\Pr[c(F_0) \geq (6(\gamma + 5)\log n) / \epsilon] \leq O(1/n^\gamma)$$

We next prove the second part of the theorem. The event $\mathcal{A}$ that is a function of the randomness of types 1, 2, and 3 is the conjunction of the following three events:
\begin{itemize}
    \item $\mathcal{A}_1$: For each $P \in \mathcal{P}_{hit}$, there exists an edge $e \in P$, such that $e \in F_0$.
    \item $\mathcal{A}_2$: $h(G, F^*)\leq (1+\epsilon)\mathbb{E}[h(G, F^*)]$.
    \item $\mathcal{A}_3$: $h(G, F_0)\geq (1-\epsilon)\mathbb{E}[h(G, F_0)]$. 
\end{itemize}
We will first show that $\mathbb{E}[inf(V,E(\vec p) \setminus F_0,s) \bigm| \mathcal{A} ]\leq (1+O(\epsilon))OPT$, and then lower-bound $\Pr[\mathcal{A}]$. 

Let us first condition on $\mathcal{A}$. Consider any $j \in [N]$. By $\mathcal{A}_1$ and the definition of the set $\mathcal{P}_{hit}$, the only vertices in $(V,E_j \setminus F_0)$ that are reachable from $s$ can be those in $V \setminus S(j)$; these vertices are exactly the ones getting infected in the $j$-th sample. Further, by definition we have $y_{vj} < \epsilon$ for every $v \in V \setminus S(j)$. Therefore, the empirical number of infections over all the samples is:
\begin{align}
h(G, F_0) &\leq \frac{1}{N}\sum_{j \in [N]} \sum_{v\not\in S(j)} 1  \leq \frac{1}{N}\sum_{j \in [N]} \sum_{v\not\in S(j)} \frac{1-y_{vj}}{1-\epsilon} \leq \frac{h(G, \hat{F})}{1-\epsilon} \leq \frac{h(G, F^*)}{1-\epsilon} \label{aux-l1}
\end{align}
The second inequality above follows because the LP value is a lower bound on $h(G, \hat{F})$, and the last inequality follows since $\hat{F}$ minimizes $h(G, F)$. Combining (\ref{aux-l1}) and the definitions of $\mathcal{A}_2$, $\mathcal{A}_3$ yields
 \begin{align}
\mathbb{E}[h(G, F_0)] \leq \frac{h(G, F_0)}{1-\epsilon} \leq \frac{h(G, F^*)}{(1-\epsilon)^2} \leq\frac{(1+\epsilon)}{(1-\epsilon)^2}\mathbb{E}[h(G, F^*)] = (1+O(\epsilon)) \mathbb{E}[h(G, F^*)]\notag
\end{align}

To conclude the proof we need to lower-bound $\Pr[\mathcal{A}]$. First, Lemma~\ref{lemma:conc} shows that each of $\mathcal{A}_2$ and $\mathcal{A}_3$ holds with probability at least $1-1/n^2$.
Let us consider $\mathcal{A}_1$ next.

Let $\mathcal{B}$ be a random variable denoting the number of paths over all the samples $G_j$. Since $\Gamma$ is the expected number of paths in a single graph, linearity of expectation gives $\mathbb{E}[\mathcal{B}] = \Gamma N = O( \frac{n^3 \Gamma}{\epsilon^2})$, since $m = O(n^2)$. Thus, by using Markov's inequality we have $\Pr[\mathcal{B} = \Omega(\frac{n^5 \Gamma}{\epsilon^2})] \leq O(1/n^2)$. Equivalently, $\Pr[\mathcal{B} = O(\frac{n^5 \Gamma}{\epsilon^2})] \geq 1-O(1/n^2)$. The randomness in the previous statements is of types 1 and 3.

Consider now a path $P \in \mathcal{P}_{hit}$. If there exists an $e \in P$ such that $x'_e = 1$, then this path is broken. Hence, assume that for all $e \in P$ we have $x'_e < 1$. By the definition of the paths in $\mathcal{P}_{hit}$ we also have $\sum_{e \in P}x'_e \geq (\gamma+5)\log n$. Therefore, the probability that all edges of $P$ survive is at most 
$$\prod_{e \in P} (1-x'_e) \leq e^{-\sum_{e \in P} x'_e} \leq e^{- (\gamma + 5) \log n } \leq n^{-(\gamma+5)}$$ In the end, a union bound over all $P \in \mathcal{P}_{hit}$ gives:
\begin{align}
    \Pr[\mathcal{A}_1 | \mathcal{B}] \geq 1 - \frac{\mathcal{B}}{n^{\gamma+5}} \notag
\end{align}
Combining everything gives $\Pr[\mathcal{A}_1] \geq (1- O(\frac{\Gamma}{\epsilon^2 n^\gamma})) (1- O(\frac{1}{n^2})) = 1- O(\frac{\Gamma}{\epsilon^2 n^\gamma}) - O(\frac{1}{n^2})$. Hence, putting down all the lower bounds for $\mathcal{A}_1, \mathcal{A}_2$ and $\mathcal{A}_3$ yields $\Pr[\mathcal{A}] \geq 1- O(\frac{\Gamma}{\epsilon^2 n^\gamma}) - O(\frac{1}{n^2})$. 
\end{proof} 

\begin{corollary}\label{SAA-final}
When $\Gamma \leq  \poly(n)$, we have $\mathbb{E}[inf(V,E(\vec p) \setminus F_0,s)]\leq (1+O(\epsilon) + O(1/n)) OPT$, where the randomness is with respect to  Type 1 (if applicable), Type 2, and Type 3.
\end{corollary}

\begin{proof}
When $\Gamma \leq \poly(n)$, we set $\gamma$ large enough such that $O(\frac{\Gamma}{\epsilon^2 n^\gamma}) = O(\frac{1}{n^2})$. Using Theorem \ref{thm:main} we then have:
\begin{align}
    \mathbb{E}[inf(V,E(\vec p) \setminus F_0,s)] &= \mathbb{E}[inf(V,E(\vec p) \setminus F_0,s) | \mathcal{A}] \Pr[\mathcal{A}] + \mathbb{E}[inf(V,E(\vec p) \setminus F_0,s) | \bar{\mathcal{A}}] \Pr[\bar{\mathcal{A}}] \notag \\ &\leq (1+O(\epsilon))OPT + n O(1/n^2) \leq (1+O(\epsilon) + O(1/n)) OPT \notag
\end{align}
To get the first inequality we use the simply upper bound of $\mathbb{E}[inf(V,E(\vec p) \setminus F_0,s) | \bar{\mathcal{A}}] \leq n$, and for the last one we use the fact that $1 \leq OPT$ ($s$ is always infected).
\end{proof}

\section{Counting Paths in the Chung-Lu Random Graph Model}
\label{sec:path-count}

Recall the random graph model of \cite{DBLP:journals/siamdm/ChungL06}. Here we are given vertices $V$, where each vertex $v \in V$ comes with a positive integer $w_v$ indicating its expected degree in the graph. For every pair of vertices $u$ and $v$, the edge $(u,v)$ is independently included in the graph with probability $q_{u,v} = w_u w_v / \sum_{r \in V}w_r$. Furthermore, we consider a power-law model, in which $n_i$, the number of nodes of weight $i$, satisfies $n_i = \Theta(n/i^{\beta})$, where $\beta > 2$ is a given parameter. Finally, recall that $w_{max} = \max_{v}w_v$, $w_{min} = \min_{v}w_v$, and a common assumption in this setting is $w_{min} = O(1)$.

Take now any random graph $G=(V,E)$ that is produced by the above model. In that graph, we assume that a disease percolation process takes place, and this process is governed by some probability vector $\vec p$. We are interested in bounding the expected number of paths $\Gamma$ in $G(\vec p)$, where the randomness of $\Gamma$ is obviously of both Types 1 and 3. To do so, we start by analyzing the expected number of paths of length $k$ in $G$, where the randomness here is only of Type 1. In what follows, we are using $\ell_k$ to denote the latter quantity.

Our first result is showing that when $\beta > 3$, we have $\ell_k \leq \poly(n, 2^k)$. Furthermore, if $p_e \leq c_0$ for all $e \in E$, where $c_0$ is a universal positive constant, we demonstrate how to utilize the bound on $\ell_k$ and eventually give a polynomial bound on $\Gamma$. 


In addition, when $\beta < 3$ we provide a negative result, indicating that our SAA framework from Section \ref{sec:SAA} cannot be utilized for this case, as no polynomial bound on $\Gamma$ can be guaranteed.

By an abuse of notation, we will let $m$ denote the \textbf{expected} number (not actual number) of edges in the graph $G$. Trivially, $m = \sum_{v \in V}w_v/2$. Since $\beta > 2$, $m$ can also be expressed as:
\begin{align}
  m= \Theta(\sum_i i \cdot n_i) = \Theta\Big(\int_{w_{min}}^{w_{max}} \frac{n}{z^{\beta-1}} dz\Big) = \Theta\Big(\frac{n}{w_{min}^{\beta-2}}\Big) \label{aux-m-bnd} 
\end{align}
The following lemma is required for counting paths.

\begin{lemma}
\label{lemma:SNk}
Fix some path length $k$, and suppose that we are given a positive integer $D \geq w_{min}$. Let also $S(D,k) \doteq \{(a(w_{min}), a(w_{min} + 1), \ldots, a(D)):~(\forall i, ~a(i) \in \mathbb{Z}_{\geq 0})
\mbox{ and } \sum_i a(i) = k\}$\footnote{$S(D,k)$ can be interpreted as the set of all vectors $x \in \mathbb{Z}^{D-w_{min} + 1}_{\geq 0}$ with $\ell^1$ norm equal to $k$. Our notation can be thought of as re-indexing $x$, so that the numbering starts at $w_{min}$ instead of $1$, and then finishes at $D$.}. Then,
\[
\ell_k \leq n \cdot \left(\frac{2^k k!}{m^k} \right) \cdot \sum_{\mathbf{a} \in S(w_{max}, k)} ~\prod_{i = w_{min}}^{w_{max}} \left( {n_i \choose a(i)} \cdot i^{2a(i)} \right). 
\]
\end{lemma}
\begin{proof}
We say that a vertex $v$ is in \emph{class} $i$ if $w_v = i$.
Fix now a vertex $v_0$. We will next upper bound how many different paths of length $k$ can start from $v_0$. We can construct such a possible path $P = (u_0, u_1, \ldots, u_k)$ as follows:
\begin{enumerate}
    \item Pick a vector $\mathbf{a} = (a(w_{min}), a(w_{min} + 1), \ldots, a(w_{max}))$ from $S(w_{max}, k)$. This will give us the selection of how many vertices from each degree class we should pick, such that in total we have chosen $k$ vertices for $P$.
    \item For the chosen $\mathbf{a}$, pick $a(i)$ vertices from each degree class $i$, where $w_{min} \leq i \leq w_{max}$. This is possible only if $a(i) \leq n_i$ for each class $i$. However, since we are only computing an upper bound, we will assume that such a selection is always possible. Notice that there may be some additional double counting, because we may end up choosing $v_0$ again. Nonetheless, as we are only concerned with an upper bound, we will permit such ``unecessary'' cases in our counting.
    \item Choose the positions of the chosen $k$ vertices among the indices $\{1, 2, \ldots, k\}$ of the path $P$ to be constructed ($k!$ possibilities).  
    \end{enumerate}
    

Overall, based on the 3 cases above, the total number of paths starting from $v_0$ can be at most:
\begin{align}
    k! \sum_{\mathbf{a} \in S(w_{max}, k)} ~\prod_{i = w_{min}}^{w_{max}} {n_i \choose a(i)}  \label{aux-cnt-1}
\end{align}

Let $d_0$ be the class of $u_0$. Suppose we complete the three steps above and let $(d_0, d_1, \ldots, d_k)$ be the ordered degree sequence obtained for the vertices in $P$, when the chosen vector was $\mathbf{a}$. The probability of such a path materializing in the edge selection phase is 
\begin{align}
\prod_{i=0}^{k-1} \left(\frac{d_i d_{i+1}}{\sum_v w_v} \right)  = \frac{2^k}{m^k} {\prod^{w_{max}}_{i = w_{min}} i^{2a(i)}} \label{aux-cnt-2}
\end{align}

Combining (\ref{aux-cnt-1}) and (\ref{aux-cnt-2}), the expected number of paths of length $k$ starting at $u_0$ is at most
\begin{align}
\left(\frac{2^k k!}{m^k} \right) \cdot \sum_{\mathbf{a} \in S(w_{max}, k)} ~\prod_{i = w_{min}}^{w_{max}} \left( {n_i \choose a(i)} \cdot i^{2a(i)} \right)
\end{align}
Summing this bound over all possible starting vertices results in the claim of the Lemma.
\end{proof}

\subsection{A Positive Result When $\beta > 3$}

We begin this section with a couple of important technical lemmas, and then move on to our final result regarding the expected number of paths in a randomly drawn graph.

\begin{lemma}
\label{lemma:beta>3}
Let $S(D, k)$ be as  in Lemma~\ref{lemma:SNk}, and let
$N(D,k) \doteq \sum_{\mathbf{a} \in S(D, k)} ~\prod_{i = w_{min}}^{D} \frac{1}{i^{c_1 a(i)} \cdot a(i)!}$ for some constant $c_1 > 1$. Then $N(D,k) \leq \frac{1}{k!} \cdot \prod_{i = w_{min} + 1}^D ~\left(1 + \frac{1}{i^{c_1}} \right)^k$. 
\end{lemma}
\begin{proof}
We prove by induction on $D$ the stronger statement (\ref{stmt-ind}), which directly implies the Lemma. 
\begin{align}
    \forall k, ~
N(D,k) \leq \frac{1}{k!} \cdot \prod_{i = w_{min} + 1}^D \left(1 + \frac{1}{i^{c_1}} \right)^k \tag{A} \label{stmt-ind}
\end{align}

The base case $D = w_{min}$ is easy. For this notice that $S(w_{min},k) =\{(k)\}$, and hence we have
$$N(w_{min}, k) = \frac{1}{w_{min}^{c_1 k} \cdot k!} \leq \frac{1}{k!}$$
The inequality above follows since $c_1 > 1$.

We complete the proof by strong induction. Suppose $D > w_{min}$. Elementary calculations and the definition of $N(\cdot, \cdot)$ reveal the following recurrence when $D > w_{min}$
\begin{equation}
\label{eqn:N-recur}    
N(D,k) = \sum_{j=0}^k \left( N(D-1, k-j) \cdot \frac{1}{D^{c_1 j} \cdot j!} \right)
\end{equation}
Recurrence (\ref{eqn:N-recur}) and the induction hypothesis yield
\begin{align}
N(D,k) &\leq  \sum_{j=0}^k \left( \frac{1}{D^{c_1 j} \cdot j!} \cdot \frac{1}{(k - j)!} \prod_{i = w_{min} + 1}^{D-1} \left(1 + \frac{1}{i^{c_1}} \right)^{k-j} \right) \notag \\
&\leq \sum_{j=0}^k \left( \frac{1}{D^{c_1 j} \cdot j!} \cdot \frac{1}{(k - j)!}  \prod_{i = w_{min} + 1}^{D-1} \left(1 + \frac{1}{i^{c_1}} \right)^{k} \right) \notag \\
&=  \left( \prod_{i = w_{min} + 1}^{D-1} \left (1 + \frac{1}{i^{c_1}} \right)^{k} \right) \sum_{j=0}^k \left( \frac{1}{D^{c_1 j} \cdot j!} \cdot \frac{1}{(k - j)!}  \right) \notag \\
&=  \frac{1}{k!} \cdot \left( \prod_{i = w_{min} + 1}^{D-1} \left (1 + \frac{1}{i^{c_1}} \right)^{k} \right) \cdot \sum_{j=0}^k \left({k \choose j} \cdot \frac{1}{D^{c_1 j}} \right) \notag \\
&=  \frac{1}{k!} \cdot \left( \prod_{i = w_{min} + 1}^{D} \left (1 + \frac{1}{i^{c_1}} \right)^{k} \right)  \notag
\end{align}
The last inequality above follows from the binomial sum $\sum_{j=0}^k \left({k \choose j} \cdot \frac{1}{D^{c_1 j}} \right) = (1+ 1/D^{c_1})^k$.
\end{proof}

\begin{lemma}
\label{lemma:ell_k}
Suppose $\beta = 2 + c_1$ for some constant $c_1 > 1$. Then, for all $k$, $\ell_k \leq \poly(n, 2^k)$. 
\end{lemma}
\begin{proof}
Before we proceed to our main arguments, we make some useful observations and give a bit more notation. At first, using (\ref{aux-m-bnd}) and the assumption that $w_{min} = O(1)$, we see that $\frac{n}{m} = O(1)$. Furthermore, because $n_i = \Theta(\frac{n}{i^\beta})$ for every $i \in [w_{min}, w_{max}]$, let $\lambda$ be a universal constant such that $n_i \leq \frac{\lambda n}{i^\beta}$ for every $i$. Using Lemma \ref{lemma:SNk} we get
\begin{eqnarray*}
\ell_k & \leq & n \cdot \left(\frac{2^k k!}{m^k} \right) \cdot \sum_{\mathbf{a} \in S(w_{max}, k)} ~\prod_{i = w_{min}}^{w_{max}} \left( \frac{n_i^{a(i)}}{a(i)!} \cdot i^{2a(i)} \right) \\
& \leq & n \cdot \left(\frac{2^k k!}{m^k} \right) \cdot \sum_{\mathbf{a} \in S(w_{max}, k)} ~\prod_{i = w_{min}}^{w_{max}} \left( \frac{(\lambda n/i^{2 + c_1})^{a(i)}}{a(i)!} \cdot i^{2a(i)} \right) \\
& \leq & n \cdot k! \cdot \left(\frac{\lambda n}{m}\right)^k \cdot \sum_{\mathbf{a} \in S(w_{max}, k)} ~\prod_{i = w_{min}}^{w_{max}} \left( \frac{(1/i^{2 + c_1})^{a(i)}}{a(i)!} \cdot i^{2a(i)} \right) \\
& = & n \cdot k! \cdot \left(\frac{\lambda n}{m}\right)^k \cdot \sum_{\mathbf{a} \in S(w_{max}, k)} ~\prod_{i = w_{min}}^{w_{max}} \frac{1}{i^{c_1 a(i)} \cdot a(i)!} \\
& = & \poly(n, 2^k) \cdot k! \cdot N(w_{max},k).
\end{eqnarray*}

Using the bound on $N(D, k)$ from Lemma~\ref{lemma:beta>3}, we have
\begin{align}
\ell_k & \leq \poly(n, 2^k) \cdot \prod_{i = w_{min} + 1}^{\infty} \left (1 + \frac{1}{i^{c_1}} \right)^{k}  \leq\poly(n, 2^k) \cdot \prod_{i = w_{min} + 1}^{\infty} e^{\frac{k}{i^{c_1}}} \notag \\ & = \poly(n, 2^k) \cdot e^{k \cdot \sum_{i > w_{min}} \frac{1}{i^{c_1}}} \leq \poly(n, 2^k) \notag
\end{align}
The last inequality follows because $\sum_{i > w_{min}} (1/i^{c_1}) = O(1)$ when $c_1 > 1$. 
\end{proof}

\begin{corollary}\label{CL-bound}
Let $G$ be a graph drawn from the Chung-Lu distribution with power law weights, with parameter $\beta=2+c_1$ for some constant $c_1>1$. Then there is a constant $c_0 > 0$ that depends only on $c_1$, such that the following holds: if the probability $p_e$ of retaining edge $e$ during the disease percolation process satisfies $p_e \leq c_0$ for any $e$, then the expected number $\Gamma$ of paths in $G(\vec p)$ is upper-bounded by  $\poly(n)$. (This expectation is over randomness of types 1 and 3.) 
\end{corollary}
\begin{proof}
From Lemma~\ref{lemma:ell_k}, we have $\ell_k\leq \poly(n, 2^k)$; let $C$ be a constant such that $\ell_k\leq n^C 2^{Ck}$ for all $k$. We choose $c_0 = 2^{-C}$. Then, when $p_e \leq c_0$ for every $e$, the probability that a given path of length $k$ survives in $G(\vec p)$ is at most $c_0^k$. Therefore,
the expected number of paths in $G(\vec p)$ is
\begin{align}
\Gamma \leq \sum_k \ell_k c_0^k\leq \sum_k n^C (c_0 2^C)^k \leq n^{C+1} \notag &\qedhere
\end{align}
\end{proof}

Combining Corollaries \ref{CL-bound} and \ref{SAA-final}, we get a bicriteria approximation algorithm for \mininfprobCL{}.

\subsection{A Negative Result When $\beta < 3$}

We now consider the case $\beta<3$ and show an interesting contrast to Lemma~\ref{lemma:ell_k}.

\begin{lemma}\label{CL-neg}
When $\beta = 2 + c_0$ for some constant $c_0 < 1$, there may exist $k$ with $\ell_k=\omega(poly(n, 2^k))$.
\end{lemma}
\begin{proof}
In the proof of Lemma \ref{lemma:SNk} we gave an upper bound for $\ell_k$. However, the double counting or the unnecessary cases we involved in our counting can only account for low-order terms. In other words, we can assume 
\begin{align}
    \ell_k = \Theta \left( n \cdot \left(\frac{2^k k!}{m^k} \right) \cdot \sum_{\mathbf{a} \in S(w_{max}, k)} ~\prod_{i = w_{min}}^{w_{max}} \left( {n_i \choose a(i)} \cdot i^{2a(i)} \right) \right) \label{aux-lst}
\end{align}

Consider the case where $w_{min} = 1$ and $w_{max} = k$, and just take the one sequence $\mathbf{a} = (1, 1, \ldots, 1)$. Furthermore, because $n_i = \Theta(\frac{n}{i^\beta})$ for every $i \in [w_{min}, w_{max}]$, let $\lambda$ be a universal constant such that $n_i \geq \frac{\lambda n}{i^\beta}$ for every $i$. Since $a(i) = 1$ for all $i$ here, the quantity inside the $\Theta$ notation in (\ref{aux-lst}), which we denote by $Q$, can be lower-bounded as follows
\begin{align}
Q & \geq n \cdot \left( \frac{2^k k!}{m^k} \right) \cdot \prod_{i = w_{min}}^{w_{max}} \left( n_i \cdot i^{2} \right)  \geq n \cdot \left(\frac{2^k k!}{m^k} \right) \cdot \prod_{i = 1}^{k} \left( (\lambda n/i^{\beta}) \cdot i^{2} \right) \notag \\
&= n \cdot \left(\frac{2^k \lambda^k n^k k!}{m^k} \right) \cdot \prod_{i = 1}^{k} i^{- c_0} = n \cdot \left(\frac{2^k \lambda^k n^k}{m^k} \right) \cdot (k!)^{1 - c_0} = \poly(n, 2^k) \cdot (k!)^{1 - c_0} \notag
\end{align}
Because $c_0 < 1$, we have that $(k!)^{1 - c_0}$ grows faster than $\poly(2^k)$. Hence, $Q = \omega(\poly(n,2^k))$ and consequently $\ell_k = \omega(\poly(n,2^k))$.
\end{proof}

Using reasoning similar to that used in Corollary \ref{CL-bound}, we see that Lemma \ref{CL-neg} implies that under this stochastic regime, our proof approach cannot provide meaningful results for \mininfprobCL{}. 
Thus we see a \emph{phase transition} for the expected number of paths of any length $k$: from at least $(k!)^{\Omega(1)}$ to $\poly(n,2^k)$ at $\beta = 3$. It is an open question as to what happens when $\beta = 3$.

\section{A Deterministic Rounding SAA Approach}
\label{sec:saa-general}

In this section we revisit the SAA approach of Section \ref{sec:SAA}, and instead of a randomized rounding, we apply a simple deterministic rounding scheme. The advantage of the latter is that the success probability of the algorithm no longer relies on the value $\Gamma$. However, this comes at the expense of much worse bicriteria factors. 

Once again we are going to sample $N=\frac{3n}{\epsilon^2} \log \big{(}n^2 \cdot 2^{m+1}\big{)}$ graphs $G_j = (V,E_j)$ from $G(\vec p)$, and then construct LP (\ref{LP-1})-(\ref{LP-4}). Let $(x,y)$ be the optimal fractional solution of the LP. In this case, our returned solution will be $F_0=\{e \in E: x_e\geq \frac{1}{4n^{2/3}}\}$. Before we proceed with our analysis, let us recall some important notation from Section \ref{sec:SAA}. For any fixed $F \subseteq E$, $h(G_j,F) = inf(V,E_j \setminus F,s)$ and $h(G, F) = \frac{1}{N}\sum_{j=1}^N h(G_j, F)$. Finally, $F^*$ denotes the optimal edge set for the given instance of \mininfprob{}, and $\hat F$ denotes the optimal integral solution of LP (\ref{LP-1})-(\ref{LP-4}).

\begin{theorem}
With high probability, $F_0$ is an $(O(n^{2/3}), O(n^{2/3}))$-approximation for \mininfprob{}.
\end{theorem}

\begin{proof}
To begin with, by the definition of $F_0$ and constraint (\ref{LP-3}), we have
$$\sum_{e \in F_0} c_e \leq 4n^{2/3}\sum_{e \in F_0} c_e x_e \leq 4 n^{2/3} B$$

Moving forward, note that by Lemma \ref{lemma:conc}, we have $h(G, F^*)\leq (1+\epsilon)\mathbb{E}[h(G, F^*)]$ and $h(G, F_0)\geq (1-\epsilon)\mathbb{E}[h(G, F_0)]$ with probability at least $1-O(1/n^2)$. If we show that $h(G,F_0) \leq 2 n^{2/3} h(G,\hat F)$, then we are done. This is because:
\begin{align}
    \mathbb{E}[h(G,F_0)] \leq  \frac{h(G,F_0)}{1-\epsilon} \leq \frac{2 n^{2/3}}{1-\epsilon} h(G,\hat F) \leq \frac{2 n^{2/3}}{1-\epsilon} h(G, F^*) \leq \frac{2(1+\epsilon) n^{2/3}}{1-\epsilon} \mathbb{E}[h(G, F^*)] \notag
\end{align}

At first, suppose $h(G,\hat F) > n^{1/3}$. Since $h(G_j, F_0) \leq n \leq n^{2/3} h(G,\hat F)$ for any $j$, $h(G,F_0) \leq 2 n^{2/3} h(G,\hat F)$ follows trivially through the definition of $h(G,F_0)$.

Next, suppose $h(G,\hat F) \leq n^{1/3}$. This implies $\frac{1}{N}\sum_{j\in [N]}\sum_{v \in V} (1-y_{vj}) \leq h(G,\hat F)\leq n^{1/3}$, because the optimal LP-value is a lower bound for $h(G,\hat F)$. Let now $A' = \{j \in [N]: \sum_{v \in V} (1-y_{jv}) \leq n^{2/3}\}$ and $A''= [N] \setminus A' = \{j \in [N]: \sum_{v \in V} (1-y_{vj}) > n^{2/3}\}$. The upper bound of the optimal fractional solution value then gives $|A''|\leq N/n^{1/3}$. Consider now any $j\in A'$, and let $v$ be a node such that $1-y_{vj} \leq 1/2$. We will argue below that for any path $P\in\mathcal{P}(s, v, G_j)$, there exists an edge $e\in P$ such that $e\in F_0$. This means that if $v$ is infected in $(V, E_j \setminus F_0)$, then $1-y_{vj} > 1/2$, and so $h(G_j, F_0) \leq \sum_v 2(1-y_{vj})$. Hence,
\begin{align*}
h(G, F_0) &= \frac{1}{N}\sum_{j \in A'}h(G_j, F_0) + \frac{1}{N}\sum_{j \in A''}h(G_j, F_0) \\
&\leq \frac{1}{N}\sum_{j\in A'} \sum_{v \in V} 2(1-y_{vj}) + \frac{n|A''|}{N} \\
&\leq \frac{1}{N}\sum_{j\in [N]} \sum_{v \in V} 2(1-y_{vj}) +  n^{2/3} \\
&\leq \frac{1}{N}\sum_{j\in [N]} \sum_{v \in V} 2(1-y_{vj}) +  n^{2/3} h(G,\hat F) \\
&\leq (2 + n^{2/3})h(G,\hat F) \leq 2 n^{2/3}h(G,\hat F)
\end{align*}
where the third inequality follows because $h(G_j, \hat F) \geq 1$, and thus $h(G, \hat F) \geq 1$.

Finally, we prove that for any $j\in A'$, and any $v$ such that $1-y_{vj}\leq 1/2$, it must be the case that for each $P\in\mathcal{P}(s, v, G_j)$ we have $P\cap F_0\neq\emptyset$. Let $P=(v_0,v_1,\ldots,v_r)$ with $v_0=s$ be such a path of $\mathcal{P}(s, v, G_j)$. First, suppose $|P|=r\leq 2n^{2/3}$. Then, constraint (\ref{LP-2}) yields $\sum_{e\in P} x_e \geq 1/2$, and hence there exists $e\in P$ with $x_e \geq 1/(2|P|) \geq 1/(4n^{2/3})$, which implies $e\in F_0$. Next, suppose $|P| > 2n^{2/3}$. Let $P'=(v_0,v_1,\ldots,v_k)$ be the prefix of $P$ of length $k=2n^{2/3}$. We will show that $P'\cap F_0\neq\emptyset$, which implies $P\cap F_0\neq\emptyset$. By definition of $A'$, we have $\sum_{i=0}^k (1-y_{v_ij}) \leq n^{2/3}$. Since $k= 2n^{2/3}$, there exists some $1\leq \ell\leq k$ such that $1-y_{v_{\ell}j} \leq 1/2$, or $y_{v_{\ell}j} \geq 1/2$. Constraint~(\ref{LP-2}) applied for $v_{\ell}$ and the path $(v_0, v_1, \hdots, v_\ell)$ gives $\sum_{i=0}^{\ell-1} x_{(v_i, v_{i+1})} \geq y_{v_{\ell}j} \geq 1/2$, and thus there exists an edge $(v_i, v_{i+1})\in P'$ with $x_{(v_i, v_{i+1})} \geq \frac{1}{2\ell}\geq \frac{1}{2k} \geq \frac{1}{4n^{2/3}}$, which means that $(v_i, v_{i+1})\in F_0$.
\end{proof}

\section{Conclusions and Future Work} 


Despite the fundamental nature of \mininfprob{} and \mininfprobnode{}, their computational complexity remained open for the $p<1$ setting. A number of heuristics have been proposed, and rigorous algorithms are only known for very special random graphs. We present the first rigorous approximation results for these problems for certain classes of instances; however, even these turn out to be quite challenging, and require adapting the cut sparsification and sample-average approximation techniques in a non-trivial manner.
Our work raises several interesting questions. First, it would be interesting to extend the result based on Karger's cut sparsification technique to the non-uniform probability setting. Second, it would be interesting to extend our work to other realistic random models of social-contact networks, and to also identify what reasonable assumptions on \emph{deterministic} network models would guarantee efficient solutions. Third, to capture a wider variety of realistic scenarios, it would be beneficial improving the constant upper bound on $p$ in Lemma \ref{CL-bound}. Finally, it is of interest to see if our approximation guarantees and running times can be improved.

\section*{Acknowledgements}

Michael Dinitz was supported by NSF award CCF-1909111. Aravind Srinivasan was supported in part by NSF awards CCF-1422569, CCF-1749864, and CCF-1918749, as well as research awards from Adobe, Amazon, and Google. Leonidas Tsepenekas was  supported in part by NSF awards CCF-1749864 and CCF-1918749, and by research awards from Amazon and Google. Anil Vullikanti's work was partially supported by NSF awards IIS-1931628, CCF-1918656, and IIS-1955797,  NIH award R01GM109718, and CDC U01CK000589.
The contents are those of the author and do not necessarily represent the official views of, nor an endorsement, by CDC/HHS, or the U.S. Government.

\printbibliography

\appendix

\section{The \normalfont{\mininfprobnode{}} \textbf{Problem}}
\label{sec:vacc}

While the results based on Karger's technique (Section~\ref{sec:karger}) do not easily extend to the \mininfprobnode{} problem, our SAA based results do, and we explain this here.

We make a few small changes to the linear program LP (\ref{LP-1})-(\ref{LP-4}). At first, we use variable $x_v$ as the indicator for removing (i.e., vaccinating) vertex $v$. Furthermore, each $P \in \mathcal{P}(s, v, G_j)$ will now contain the vertices of the path and not the edges. Everything else remains the same, and thus we get the following linear program, denoted by $LP_{vacc}$:
\begin{align}
    \min &\frac{1}{N}\sum_j \sum_v (1-y_{vj}) \text{ such that} \label{LP-v1} \\
    &\sum_{v\in P} x_v \geq y_{vj}, ~~\forall j ~\forall P\in \mathcal{P}(s, v, G_j)  \label{LP-v2}\\
    &\sum_v c_v x_v \leq B \label{LP-v3}\\
    &x_v, y_{vj} \in [0,1], ~\text{ for all $j\in [N], v \in V$}\label{LP-44}
\end{align}

\noindent
\textbf{\mininfprobnode{} in the Chung-Lu model.}
Our rounding scheme now involves constructing a subset $F_0\subseteq V$, by picking each $v
\in V$ with probability 
\begin{align}
    x'_v=\min\Big{\{}\frac{(\gamma+5)x_v\log{n}}{\epsilon}, 1\Big{\}} \notag
\end{align}

It is easy to verify that Theorem~\ref{thm:main} and Corollary~\ref{SAA-final} hold in the case of vertex removal, by considering the quantity $inf(V,E \setminus \{(u, v)\in E: u\in F\mbox{ or }v\in F\},s)$ instead of $inf(V, E\setminus F, s)$. Corollary~\ref{CL-bound} is unchanged for the vertex removal case, as well. Putting these together, we have the following result for the \mininfprobnode{} problem.

\begin{corollary}\label{CL-vacc}
The solution $F_0$ picked by the rounding scheme above is an $(O(\log{n}), O(1))$-approximation for the \mininfprobnode{} problem for graphs drawn from the Chung-Lu model with power law weights, with parameter $\beta=2+c_1$ for some constant $c_1>1$.
\end{corollary}

\noindent
\textbf{\mininfprobnode{} in general graphs.} The deterministic rounding of Section \ref{sec:saa-general} holds if the rounding for edges is replaced by the same rounding for nodes, which gives us the following result.

\begin{corollary}
There is an $(O(n^{2/3}), O(n^{2/3}))$--approximation for the \mininfprobnode{} problem.
\end{corollary}

\section{Auxiliary Lemmas}

\begin{lemma}\label{chernoff}
\cite{chernoff} Let $X_1, X_2, \hdots, X_K$ be independent random variables with $X_k \in [0,1]$ for every $k$. For $X = \sum^K_{k=1}X_k$ with $\mu = \mathbb{E}[X]$:
\begin{itemize}
    \item For any $\delta > 0$, we have $\Pr\big{[}X \notin [(1-\delta)\mu, (1+\delta)\mu]\big{]} \leq e^{\frac{-\mu \delta^2}{3}}$.
    \item For any $R \geq 6 \mu$, we have $\Pr[X \geq R] \leq 2^{-R}$.
\end{itemize}
\end{lemma}

\end{document}